\documentclass[letterpaper,11pt]{article}
\usepackage{amssymb}
\usepackage{amsfonts}
\usepackage{amssymb,amsmath,amsthm}

\usepackage{subfig}
\usepackage{graphicx}

\usepackage[ruled]{algorithm}
\usepackage[noend]{algpseudocode}

\usepackage[letterpaper,nohead,margin={1in,1.2in}]{geometry}

\usepackage{setspace}
\newtheorem{theorem}{Theorem}[section]

\newtheorem{lemma}[theorem]{Lemma}

\theoremstyle{definition}
\newtheorem{definition}[theorem]{Definition}

\theoremstyle{remark}

\numberwithin{equation}{section}

\date{September 23, 2011}
\begin{document}

\title{Computing on Binary Strings}

\author{
Tian-Ming Bu\thanks{Shanghai Key Laboratory of Trustworthy Computing, East
China Normal University, Shanghai, P.R.China. Email:
\texttt{tmbu@sei.ecnu.edu.cn}\,.} \and Chen Yuan\thanks{School of Computer
Science, Fudan University, Shanghai, P.R.China. Email:
\texttt{ych04@hotmail.com}\,.}
\and Peng Zhang\thanks{Shanghai Key Laboratory of Trustworthy Computing, East
China Normal University, Shanghai, P.R.China. Email:
\texttt{arena.zp@gmail.com}\,.} }

\maketitle

\begin{abstract}
Many problems in Computer Science can be abstracted to the following question:
given a set of objects and rules respectively, which new objects can be
produced? In the paper, we consider a succinct version of the question: given a
set of binary strings and several operations like \emph{conjunction} and
\emph{disjunction}, which new binary strings can be generated? Although it is a
fundamental problem, to the best of our knowledge, the problem hasn't been
studied yet. In this paper, an $O(m^2 n)$ algorithm is presented to determine
whether a string $s$ is representable by a set $W$, where $n$ is the number of
strings in $W$ and each string has the same length $m$. However, looking for
the minimum subset from a set to represent a given string is shown to be
$\mathcal{NP}$-hard. Also, finding the smallest subset from a set to represent each string in the original set is $\mathcal{NP}$-hard. We establishes inapproximability results and approximation algorithms for them. In addition, we prove that counting the number of strings
representable is $\#\mathcal{P}$-complete. We then explore how the problems change when the operator \emph{negation} is available. For example, if the operator \emph{negation}
can be used, the number is some power of 2. This difference maybe help us
understand the problem more profoundly.
\end{abstract}

\textbf{classification}: {algorithm design}

\newpage

\section{Introduction}

The 24 Game is a popular card game in which the players randomly pick up 4
cards, then try to get 24 from the numbers on the cards through addition,
subtraction, multiplication, or division. The idea behind the game may be
abstracted to the following question: given a set of objects and rules
respectively, which new objects can be produced? In Computer Science or other
disciplines, many problems actually have the same idea. For example, in the
proof theory, a fundamental question is to decide whether a proposition can be
deduced from a given axiom system. The \emph{subset sum problem}
\cite{karp1972}, deciding whether a specific number is a sum of a subset of
the given integers, is another example of this kind, which has been proved
$\mathcal{NP}$-complete.

In this paper, we consider a succinct version of the question: given a set of
binary strings and several operations like \emph{conjunction} and
\emph{disjunction}, which new binary strings can be generated? The problem can
also be described in the language of set theory. Specifically, given an
universal set and some subsets, which new sets can be generated by
\emph{intersection} and \emph{union} operations? Clearly, the problem is
intrinsic enough to have many theoretical and practical applications.

\subsection{Our Contributions}

\textbf{Decision problem}: For the decision problem to determine whether a
string $s$ can be generated from a given set of strings $W$ by a formula with
operators \emph{disjunction} and \emph{conjunction}, an $O(m^2 n)$ algorithm is
present where $n$ is the number of each string in $W$ and $m$ is the length of
strings in $W$. If the operator \emph{negation} is allowed, the algorithm still
works.

\textbf{Optimization problem}: Whether the operator \emph{negation} is allowed
or not, we prove that looking for the minimum subset from a set to represent a
given string is $\mathcal{NP}$-hard by reducing the \emph{minimum set cover
problem} to it. Furthermore, an approximation algorithm is given through an
approximation preserving reduction to the \emph{minimum set cover problem}. Besides, we studied finding the smallest subset to have the same set of representable strings as the original set in section \ref{sec:decide_MSS}. This also showed to be $\mathcal{NP}$-hard and is asymptotically as hard to approximate as minimum set cover problem.

\textbf{Counting problem}: We prove that counting the number of strings
representable by $W$ through operators \emph{disjunction} and
\emph{conjunction} is $\#\mathcal{P}$-complete through reducing the problem of
counting the number of antichains to the problem of counting the number of
upper sets to this problem. In addition, if the operator \emph{negation} can be
used, based on the reduction, we show that the number equals $2^{|U|}$, where
$U$ is the set of equivalent classes derived from $W\cup \overline W$.

\subsection{Related Work}

To the best of our knowledge, the most related topic with this problem is
Boolean algebra. However, the \emph{negation} operator which is indispensable
in Boolean algebra, is not considered in this paper except for the last section.
So no results in Boolean algebra can be used. In the last section, one of the
studied problems is to count the number of strings representable by operators
\emph{disjunction}, \emph{conjunction} and \emph{negation} from some initial
set $W$. Because the set of generated strings constitutes a Boolean algebra,
due to the \emph{representation theorem} by M. H. Stone in \cite{Stone1936}
that every abstract Boolean algebra can be interpreted as a Boolean algebra of
all subsets of some specially chosen universal set, and vice verse, the number
of generated strings is in the form of power of 2. Our result (Theorem
\ref{thm:count_negation}) further improves the theorem because just from the
initial set we can give the \emph{exact} number of strings representable.

\bigskip

The paper is organized as follows. Section 2 gives a more precise description
of the problem. Some necessary notations are also introduced. In Section 3, the
algorithm of deciding whether a string is representable is described in
details. In Section 4, we prove that counting the number of strings
representable is $\#\mathcal{P}$-complete. Section 5 studies the optimization
problem of looking for the minimum representation subset. The last section
discusses the problems where the operator \emph{negation} is allowed.

\section{Notations}
Given two binary strings with the same length, namely $s_1, s_2$, let $s_1 \wedge
s_2$ (resp. $s_1 \vee s_2$) be the binary string produced by bitwise AND $\wedge$
(resp. OR $\vee$) of $s_1$ and $s_2$. Given a set of $m$ bits long binary strings,
namely, $W=\{s_1,s_2, \cdots, s_n\}$, $s_i \in \{0,1\}^m$, if there is a
formula $\phi$ which calculates $s$, with operators in $\{ \wedge, \vee \}$ and
operands in some subset of $W$, then we say the target string $s$ is
representable by $W$ via formula $\phi$, or simply $s$ is representable. The
\emph{binary-string representability problem} (BSR), is to decide
whether a binary string $s$ is representable by a given string set $W$ or not.

Let $x$ denote any binary string, $b_{i}^{x}$ denote the $i^{\text{th}}$ bit of
$x$. So, $x=b_{1}^{x}b_{2}^{x} \cdots b_{m}^{x}$. Also, we define a function
$\mathsf{Zero}: \mathsf{Zero}(x)=\{i|b_{i}^{x}=0\}$, from a binary string to a
set of natural numbers which denotes the indices of bits with value 0 in the
binary string. Similarly, $\mathsf{One}(x)$ denotes the indices of 1 valued
bits of $x$. Also, $\mathbf{0}$ (resp. $\mathbf{1}$) denotes a binary string
with no 1 (resp. 0) valued bits. That is,
$\mathsf{One}(\mathbf{0})=\mathsf{Zero}(\mathbf{1})=\emptyset$.

If all strings in the set $W$ have the same value in some bit, obviously the
generated string must have the same value in the same bit whatever the
generation formula $\phi$ is. So, without loss of generality, it is justifiable to
assume
$\bigcap_{x \in W}\mathsf{One}(x)=\emptyset$ and $\bigcap_{x \in
W}\mathsf{Zero}(x)=\emptyset$ respectively.

In addition, $T_i$ denotes the set of binary strings in $W$ whose
$i^{\text{th}}$ bit value is 0, i.e., ${T_i} = \{x \in W|b^x_i=0\}$. Let ${t_i}
= \bigvee_{x \in {T_i}} x$.

\section{Binary-string Representability Problem}\label{sec:decide_s}

In this section, we will present an $O(m^2 n)$ algorithm to solve the
binary-string representability problem.

\begin{algorithm}
\caption{Given $(W,s)$, determine whether $s$ is representable}
\begin{algorithmic}[1]
\Function{\textsc{Binary-String-Representability}}{$W,s$}

\If{$s=\mathbf{1}$}

   \If{($\bigvee_{x \in W} {x} = \mathbf{1}$)}

   \State \textbf{return} \textbf{TRUE}

   \Else

   \State \textbf{return} \textbf{FALSE}

   \EndIf

\EndIf

\ForAll{$i \in \mathbf{Zero}(s)$}

\State compute $t_i$

\EndFor

\If{($\bigwedge_{i \in \mathsf{Zero}(s)} {t_i} = s$)}

\State \textbf{return} \textbf{TRUE}

\Else

\State \textbf{return} \textbf{FALSE}

\EndIf

\EndFunction
\end{algorithmic}
\end{algorithm}

Next we prove the correctness of the algorithm.

\begin{lemma}\label{lem:CNFformula}
If there is a formula $\phi$ for $s$, there is an equivalent CNF $\phi_{CNF}$
for $s$, and any operands in $\phi_{CNF}$ is also in $\phi$.
\end{lemma}

\begin{proof}
We prove it by induction on the number of operators in $\phi$. If $\phi$ has no
operators, it is a CNF. Now suppose each $\phi$ with less than $n$ operands has an
equivalent CNF. Given a $\phi$ with $n$ operands, if the last operand is
$\wedge$, namely $\phi=\phi_1 \wedge \phi_2$ where both $\phi_1$ and $\phi_2$
have less than $n$ operands, then $\phi$ has an equivalent CNF since both
$\phi_1$ and $\phi_2$ have equivalent CNFs respectively. If the last operand is
$\vee$, namely $\phi=\phi_1 \vee \phi_2$, first we can write $\phi=\phi'_1 \vee
\phi'_2$ where $\phi'_1$ and $\phi'_2$ are equivalent CNFs of $\phi_1$ and
$\phi_2$ respectively. Then by applying distributivity iteratively, it is easy
to see $\phi_{CNF}=\bigwedge_{i,j}{(c_{1i} \vee c_{2j})}$ where $c_{1i}$ and
$c_{2j}$ are disjunctive clauses of $\phi'_1$ and $\phi'_2$ respectively.
\end{proof}

\begin{theorem}\label{lem:basic}
Given $(W,s)$ where $s \neq \mathbf{1}$, the following three propositions are
equivalent.
\begin{enumerate}
   \item $s$ is representable by $W$.

   \item $\forall i \in \mathsf{Zero}(s)$, $\mathsf{One}(s) \subseteq
       \mathsf{One}(t_i)$.

   \item $\bigwedge_{i \in \mathsf{Zero}(s)} {t_i} = s$.
\end{enumerate}
\end{theorem}

\begin{proof}
$(1)\Rightarrow (2)$: If $s$ is representable by $W$, according to the lemma \ref{lem:CNFformula}, there is an equivalent CNF formula $\phi_{CNF}$ for $s$. For any clause
$c$ of $\phi_{CNF}$, $\mathsf{One}(s) \subseteq \mathsf{One}(c)$, otherwise,
$x$ cannot be generated from those conjuncts. For any $i \in
\mathsf{Zero}(s)$, clearly there is at least one clause $c$ of $\phi_{CNF}$
with $ i \in \mathsf{Zero}(c)$. According to $t_i$'s definition, this clause
$c$ satisfies $\mathsf{One}(c) \subseteq \mathsf{One}(t_i)$. So,
$\mathsf{One}(s) \subseteq \mathsf{One}(c) \subseteq \mathsf{One}(t_i)$.

$(2)\Rightarrow (3)$: If $\forall i \in \mathsf{Zero}(s)$, $\mathsf{One}(s)
\subseteq \mathsf{One}(t_i)$, then $\mathsf{One}(s) \subseteq \bigcap_{i\in
\mathsf{Zero}(s)}\mathsf{One}(t_i)$. On the other hand,
$\overline{\mathsf{One}(s)}=\mathsf{Zero}(s) \subseteq \bigcup_{i\in
\mathsf{Zero}(s)} \overline{ \mathsf{One}(t_i) } = \overline{\bigcap_{i\in
\mathsf{Zero}(s)}\mathsf{One}(t_i)}$, because $\forall i \in \mathsf{Zero}(s)$,
$i \in \mathsf{Zero}(t_i)=\overline{\mathsf{One}(t_i)}$. So, $\mathsf{One}(s)
\supseteq \bigcap_{i\in \mathsf{Zero}(s)}\mathsf{One}(t_i)$. Therefor
$\mathsf{One}(s) = \bigcap_{i\in \mathsf{Zero}(s)}\mathsf{One}(t_i)$. So $s$
can be generated by the formula $\bigwedge_{i \in \mathsf{Zero}(s)} {t_i}$.

$(3)\Rightarrow (1)$: Since $\bigwedge_{i \in \mathsf{Zero}(s)} {t_i} = s$, $s$
is representable by $W$.
\end{proof}

Because $\mathbf{1}$ is representable if and only if $\bigvee_{s_i \in W} {s_i}
= \mathbf{1}$ holds, together with Theorem \ref{lem:basic}, the algorithm is
clearly correct. Furthermore, since lines 2--6 take $O(mn)$, lines 7--8 take
$O(m^2 n)$ and lines 9--12 take $O(m^2)$, the whole running time of the
algorithm is $O(m^2 n)$.

\section{Number of Representable Strings}\label{sec:countS}

In this section, we will discuss the following counting problem: given $W$, how many
binary strings can be generated from $W$? We use \#BSR to denote this number.
By reducing the problem of counting the number of antichains to the problem, we
prove the problem is $\#\mathcal{P}$-complete. Before giving the details, we
introduce some concepts first.

Given the set $\{1,\cdots,m\}$, we define an equivalence relation $\sim$ on it,
such that $i\sim j$ if and only if $T_i=T_j$. Let $U=\{[1],\cdots,[m]\}$ where
$[i]$ is the equivalence class of $i$, represent the partition. Note
that $U$ is still a well defined set under this notation even if there may be
several equivalence classes representing the same one. We use it to avoid more
unnecessary symbols. Based on the partition $U$, we define a binary relation
$\preceq_U$ such that $[i] \preceq_U [j]$ if and only if $T_i \subseteq T_j$.

\begin{lemma}\label{lem:poset}
$(U,\preceq_U)$ is a poset (partial ordered set).
\end{lemma}

\begin{proof}
Since $T_i \subseteq T_i$, $[i]\preceq_U [i]$ (\emph{reflexivity}). If $([i]
\preceq_U [j]) \wedge ([j] \preceq_U [i])$, then $(T_i \subseteq T_j) \wedge
(T_j \subseteq T_i)$. Thus $T_i=T_j$ and $[i]=[j]$ (\emph{antisymmetry}). If
$([i] \preceq_U [j]) \wedge ([j] \preceq_U [k])$, then $(T_i \subseteq T_j)
\wedge (T_j \subseteq T_k)$. Thus $[i] \preceq_U [k]$ (\emph{transitivity}).
\end{proof}

\begin{definition}[upper set]
For a poset $(X,\preceq)$, a subset $A \subseteq X$ is an upper set if and only if
 $\forall_{a \in A}{\left(a \preceq b \rightarrow b \in A\right)}$.
\end{definition}

\begin{lemma}\label{lem:equivClass}
For any representable string $s$, $i \in \mathsf{Zero}(s)$ if and only if $[i]
\subseteq \mathsf{Zero}(s)$.
\end{lemma}

\begin{proof}
Clear, if $[i] \subseteq \mathsf{Zero}(s)$, then $i \in \mathsf{Zero}(s)$.
Now suppose $i \in \mathsf{Zero}(s)$. $\forall j\in [i]$, $T_j=T_i$ and
$t_j=t_i$. So $j \in \mathsf{Zero}(t_j)=\mathsf{Zero}(t_i)$. Since $s$ is
representable and $i \in \mathsf{Zero}(s)$, according to Theorem
\ref{lem:basic}, $\mathsf{One}(s) \subseteq \mathsf{One}(t_i)$. Equivalently
$\mathsf{Zero}(t_i) \subseteq \mathsf{Zero}(s)$. Thus $j \in \mathsf{Zero}(t_i)
\subseteq \mathsf{Zero}(s)$. So $j \in \mathsf{Zero}(s)$.
\end{proof}

\begin{lemma}\label{lem:1Stringto1Set}
Given $W$, the number of representable strings is the same as the number of
upper sets of $(U, \preceq_U)$.
\end{lemma}

\begin{proof}
We will construct a bijective function from the set of representable strings to
the set of upper sets. The bijective function is defined as follows:
$\mathsf{Zero}^*(s)=\{[i] | i \in \mathsf{Zero}(s)\}$. The domain of the
function is the set of representable strings. Next we will prove the codomain
of the function is the set of upper sets, and the function is bijective.

Clearly, $\mathbf{1}$ is representable since we assume $\bigcap_{x \in
W}\mathsf{Zero}(x)=\emptyset$. $\mathsf{Zero}^*(\mathbf{1})=\emptyset$ which is
surly an upper set of $(U, \preceq_U)$. In the following proof, we will assume
$s \neq \mathbf{1}$.

For each representable string $s$, if $[i] \in \mathsf{Zero}^*(s)$ and $[i]
\preceq_U [j]$, then $T_i \subseteq T_j$. So $\mathsf{One}(t_i) \subseteq
\mathsf{One}(t_j)$ and $\mathsf{Zero}(t_i) \supseteq \mathsf{Zero}(t_j)$.
According to Theorem \ref{lem:basic}, since $s$ is representable and $i \in
\mathsf{Zero}(s)$, $\mathsf{Zero}(t_i) \subseteq \mathsf{Zero}(s)$. So
$\mathsf{Zero}(s) \supseteq \mathsf{Zero}(t_i) \supseteq \mathsf{Zero}(t_j)$.
Since $j \in \mathsf{Zero}(t_j)$, $j \in \mathsf{Zero}(s)$ and $[j] \in
\mathsf{Zero}^*(s)$. This shows $\mathsf{Zero}^*(s)$ is an upper set for each
$s$ representable by $W$.

According to the definition of the function $\mathsf{Zero}^*(\cdot)$ and Lemma
\ref{lem:equivClass}, $\bigcup_{X \in \mathsf{Zero}^*(s)}X=\mathsf{Zero}(s)$.
So if $\mathsf{Zero}^*(s_1)=\mathsf{Zero}^*(s_2)$, then
$\mathsf{Zero}(s_1)=\mathsf{Zero}(s_2)$. Thus $s_1=s_2$. Therefore
$\mathsf{Zero}^*(s)$ is \emph{injective}.

If $A$ is an nonempty upper set, we construct a string $s$ with
$\mathsf{Zero}(s)=\bigcup_{X \in A}{X}$. Clearly $\mathsf{Zero}^*(s)=A$.
$\forall i \in \mathsf{Zero}(s)$, if $j \in \mathsf{Zero}(t_i)$, then $T_i
\subseteq T_j$ and $[i] \preceq_U [j]$. Since $A$ is an upper set and $i \in
\mathsf{Zero}(s)$, $[i] \in \mathsf{Zero}^*(s)=A$ and $[j] \in A$. Because
$\mathsf{Zero}(s)=\bigcup_{X \in A}{X}$, $j \in \mathsf{Zero}(s)$.
Consequently, if $j \in \mathsf{Zero}(t_i)$, $j \in \mathsf{Zero}(s)$.
Therefore, $\forall i \in \mathsf{Zero}(s)$, $\mathsf{Zero}(t_i) \subseteq
\mathsf{Zero}(s)$. According to Theorem \ref{lem:basic}, $s$ is representable.
So the function is \emph{surjective}.
\end{proof}

\begin{definition}[antichain]
For a poset $(X,\preceq)$, a subset $A \subseteq X$ is an antichain if and only
if $\forall_{a,b \in A} (a \neq b \rightarrow (a \npreceq b \wedge b \npreceq
a))$.
\end{definition}

\begin{lemma}\label{lem:1Setto1AChain}
Given any poset $(X,\preceq)$, the number of upper sets is the same as the
number of antichains.
\end{lemma}

\begin{proof}
Let $\mathsf{Min}(A)=\{a \in A | \forall_{b \in A} (b \preceq a \rightarrow
b=a)\}$ denote the minimal elements of $A$, where $A$ is an upper set. It is
clear that $\mathsf{Min}(A)$ is an antichain. We will show that
$\mathsf{Min}(A)$ is a bijection from upper sets to antichains. For any two
different upper sets, $A_1$ and $A_2$, without loss of generality, suppose $a
\in A_2\setminus A_1$. Then there exists $a$'s predecessor $b$ such that $b \in
\mathsf{Min}(A_2)$. However, any predecessor of $a$ must not belong to $A_1$,
otherwise $a \in A_1$. So $b \notin \mathsf{Min}(A_1)$. Thus $\mathsf{Min}(A_1)
\neq \mathsf{Min}(A_2)$. So this function is \emph{injective}. If $C$ is an
antichain of $(X, \preceq)$, let $A=\{ a \in X| \exists_{b\in C} b \preceq
a\}$. Obviously $A$ is an upper set. Thus the function is \emph{surjective}.
Because $\mathsf{Min}(\cdot)$ is bijective, the number of upper sets is the
same as the number of antichains.
\end{proof}

\begin{theorem}\label{thm:CntStrings}
$\#BSR$ is $\#\mathcal{P}$-complete.
\end{theorem}

\begin{proof}
Counting antichains ($\#AC$ for short) of a poset is shown to be
$\#\mathcal{P}$-complete in \cite{provan83}. To prove the theorem above, we
construct a parsimonious reduction from $\#AC$ to $\#BSR$. Given a poset
$(P,\preceq_{P})$ where $P=\{1,\cdots,m\}$, let $W=\{s_1,\cdots,s_m\}$ where
$\mathsf{Zero}(s_i)=\{j | i\preceq_{P} j\}$. Thus $T_i=\{s_k|i \in
\mathsf{Zero}(s_k)\}=\{s_k| k \preceq_P i\}$. If $i\preceq_{P} j$, then if $s_k
\in T_i$, $s_k \in T_j$ because of the transitivity of $k \preceq_P i$ and
$i\preceq_{P} j$. So $T_i \subseteq T_j$. On the other hand, if $T_i \subseteq
T_j$, since $s_i \in T_i$, $s_i \in T_j$. So $i \preceq_P j$. Therefore
$i\preceq_{P} j \equiv T_i \subseteq T_j \equiv [i] \preceq_U [j]$. This shows
$(U,\preceq_U)$ is isomorphic to $(P,\preceq_{P})$. So the number of antichains
on $(P,\preceq_{P})$ is the same as the number of antichains on
$(U,\preceq_U)$. (The readers maybe have observed that actually the set
$U=\{\{1\},\cdots,\{m\}\}$, namely $\forall i \neq j$, $T_i \neq T_j$.) Clearly
the set $W$ can be constructed in polynomial time. Therefore, according to
Lemma \ref{lem:1Stringto1Set} and \ref{lem:1Setto1AChain}, the number of
antichains of a poset $(P,\preceq_{P})$ is the same as the number strings
generated from the set $W$. So $\#BSR$ is $\#\mathcal{P}$-complete.
\end{proof}

\section{Minimum Representation Subset Problem}\label{sec:decide_S_K}

In Section \ref{sec:decide_s}, we study the problem to decide whether a string
is representable by $W$. In this part, we hope the string can be generated by
as few strings as possible, if it is representable by $W$. In other words,
given $(W,s)$, we try to find the \emph{minimum representation subset} from $W$
to represent the string $s$. We call it the \emph{minimum representation subset
problem}.

\begin{theorem}
The minimum representation subset problem is $\mathcal{NP}$-hard.
\end{theorem}

\begin{proof}
We reduce the \emph{minimum set cover problem} to it. Given an instance
$(\mathcal{U,S})$ where $\mathcal{U}=\{1,2,\dots,m\}$ is the universe, and
$\mathcal{S}=\{S_1, \dots, S_n\}$ is a family of subsets of $\mathcal{U}$, the
minimum set cover problem is to look for the minimum subfamily $\mathcal{C}
\subseteq \mathcal{S}$ whose union is $\mathcal{U}$. We construct an instance
of the minimum representation subset problem as follows. Keep $m$ unchanged,
let $s=\mathbf{1}$, $W=\{ s_i | \mathsf{One}(s_i) = S_i\}$. Obviously, there is
a subfamily of at most $k$ sets whose union is $\mathcal{U}$ if and only if
there is a set of at most $k$ strings whose disjunction is $\mathbf{1}$.
Therefore, the minimum of $(\mathcal{U,S})$ is the same as the minimum of
$(W,\mathbf{1})$.
\end{proof}

Clearly, the reduction in the proof can also be used directly to get a $(\ln
m)$-approximation algorithm for $(W,\mathbf{1})$. For $(W,\mathbf{0})$, we can
look for the minimum representation subset of $(\overline{W},\mathbf{1})$ where
$\overline{W}$ is the set of negation of strings in $W$. Because if $s$ is
representable by some set $C$, then $\overline{s}$ is representable by
$\overline{C}$ by applying DeMorgan's laws, and vice verse.

Given $(W,s)$ where $s \notin \{\mathbf{1},\mathbf{0}\}$, in the following, we
will show a $(2\ln\frac{m}{2})$-approximation algorithm via an approximation
preserving reduction to the \emph{minimum set cover problem}. Let the
universe be the Cartesian product of $\mathsf{Zero}(s)$ and $\mathsf{One}(s)$,
i.e., $\mathcal{U}=\mathsf{Zero}(s) \times \mathsf{One}(s)$. Since $s \notin
\{\mathbf{1},\mathbf{0}\}$, $\mathcal{U} \neq \emptyset$. For each $s_i \in W$,
we create a corresponding subset $S_i \in \mathcal{S}$ such that
$S_i=\mathsf{Zero}(s_i) \times \mathsf{One}(s_i)$. For a subfamily
$\mathcal{C}$, we use $C$ to denote the corresponding set of strings.

\begin{lemma}
$\mathcal{C}$ can cover $\mathcal{U}$ if and only if $C$ can generate $s$.
\end{lemma}

\begin{proof}
If $\mathcal{C}$ covers $\mathcal{U}$, $\forall i \in \mathsf{Zero}(s)$ we
define $\mathcal{C}_i=\{X\in \mathcal{C}|X\cap (\{i\} \times \mathsf{One}(s))
\neq \emptyset\}$. Let $C_i$ denote the corresponding subset of $\mathcal{C}_i$. By the
construction, we know that $C_i \subseteq \{x\in C|b^x_i=0\}$. Since
$\mathcal{C}$ covers $\mathcal{U}$, $\forall i \in \mathsf{Zero}(s)$, $\{i\}
\times \mathsf{One}(s) \subseteq \bigcup_{X \in \mathcal{C}_i}X$. Thus,
$\forall i \in \mathsf{Zero}(s)$, $\mathsf{One}(s) \subseteq
\mathsf{One}(\bigvee _{x \in C_i}x) \subseteq \mathsf{One}(\bigvee\{x\in
C|b^x_i=0\})$. According to Theorem \ref{lem:basic}, $s$ is representable by
$C$.

Conversely, if $s$ is representable by $C$, according to Theorem
\ref{lem:basic}, $\forall i \in \mathsf{Zero}(s)$, $\mathsf{One}(s) \subseteq
\mathsf{One}(\bigvee _{x \in C_i}x)$, where $C_i = \{x\in C|b^x_i=0\}$. Let
$\mathcal{C}_i$ be the corresponding subfamily. Then $\forall i \in
\mathsf{Zero}(s)$, $\{i\} \times \mathsf{One}(s) \subseteq \bigcup_{X \in
\mathcal{C}_i}X$. So $\mathcal{C}$ covers $\mathcal{U}$.
\end{proof}

It is easy to see that the maximum cardinality of a set in $\mathcal{S}$ is no
larger than $(\frac{m}{2})^2$. Consequently, by running the greedy algorithm
\cite{Johnson:1973:AAC:800125.804034} on $(\mathcal{U},\mathcal{S})$, we get a
$(2\ln\frac{m}{2})$-approximation algorithm for the minimum representation
subset problem.

\section{Minimum Spanning Subset Problem}\label{sec:decide_MSS}

In section \ref{sec:decide_S_K}, we study how to find the minimum subset which is enough to represent a given string $s$. A natural generalization is asking for the minimum subset to represent every string in $W$, i.e., find a minimum subset $A\subseteq W$, so that each $s\in W$ is representable by $A$. We refer to this problem as \emph{Minimum Spanning Subset (MSS)}. This definition implies that, for any string $s$, $s$ is representable by $A$ if and only if $s$ is representable by $W$. In a sense, $A$ has the same power of representation as $W$.
Let $U^A$ be the counterpart of $U$ defined on $A$, and $\preceq_{U^A}$ be the counterpart of $\preceq_U$.
According to Section \ref{lem:1Stringto1Set}, $A$ has the same power of representation as $W$ if and only if $\left((U^A), \preceq_{U^A} \right)$ is equivalent as $\left( (U,\preceq_U) \right)$. To make the problem more clear, we rephrase it as a more independent problem, \emph{Minimum Compare Set}.

\textbf{Minimum Compare Set (MCS)} Given a set of items, $A=\{a_1,a_2,\cdots,a_m\}$, and a collection $B$ of subsets of $A$, $B=\{b_1,b_2,\cdots,b_n\}$. A subset $X\subseteq A$ is called a \emph{compare set} for $(A,B)$ if and only if for any two sets in $B$, say $b_i, b_j$, $b_i \subseteq b_j$ if and only if $(b_i \cap X) \subseteq (b_j \cap X)$. The decision problem of a $k$ sized \emph{compare set} is denoted as $MCS(A,B,k)$.

Each string in MSS corresponds to an item in MCS and vice versa, and each equivalence class of $U$ in MSS corresponds to each subset in MCS and vice versa. So MCS is just a reformulation of MSS, they are actually the same problem. As we are going to reduce Minimum Set Cover (MSC) problem to MCS by a similar way introduced in \cite{moret:983}, we definite MSC again as follows.

\textbf{Minimum Set Cover (MSC)} Given a set of items, $\mathcal{U}=\{u_1,u_2,\cdots,u_m\}$, and a collection of subsets of $\mathcal{U}$, $\mathcal{F}=\{f_1,f_2,\cdots,f_n\}$. A subcollection $\mathcal{C}\subseteq \mathcal{F}$ is called a \emph{set cover} for $(\mathcal{U},\mathcal{F})$ if and only if $\bigcup_{f_i \in \mathcal{C}}{f_i} = \mathcal{U}$. The decision problem of a $k$ sized \emph{set cover} for $(\mathcal{U},\mathcal{F})$ is denoted as $MSC(\mathcal{U},\mathcal{F},k)$.

\begin{theorem}
Minimum Compare Set is NP-complete.
\end{theorem}

\begin{proof}
Given any instance of MSC, e.g., $MSC(\mathcal{U},\mathcal{F},k)$, we create an instance of MCS, $MCS(A,B,|\mathcal{U}|+k)$, where $|A|=|\mathcal{U}|+|\mathcal{F}|$ and $|B|=2|\mathcal{U}|$. For simplicity, let $|\mathcal{U}|=m$, $|\mathcal{F}|=n$. In MCS, $a_1,a_2,\cdots,a_m$ correspond to sets with a single element, $\{u_1\}, \{u_2\}, \cdots, \{u_m\}$, and $a_{m+1},\cdots, a_{m+n}$ correspond to sets $f_1,\cdots,f_n$. There are $2m$ subsets in MCS, the latter $m$ subsets are $\forall_{1\le i\le m}{b_{m+i}=\{a_i\} }$, and the forgoing $m$ ones are $\forall_{1\le i\le m}{b_i= a_{i} \cup \{a_j | u_i \in f_{j-m}\} } $. This completes the polynomial transformation and is illustrated in table \ref{tab:MSC(U,F,k)} and table \ref{tab:MCS(A,B,|U|+k)}. 1 in the table stands for containment, and blank stands for non-containment. Now we claim that $MCS(A,B,|\mathcal{U}|+k)$ is YES if and only if $MSC(\mathcal{U},\mathcal{F},k)$ is YES.
For simplicity, let $LEFT=\{a_{1}, \cdots, a_{m} \}$,  $RIGHT=\{ a_{m+1}, \cdots, a_{m+n} \}$, $UP=\{b_1,\cdots,b_m\}$, $LOW=\{b_{m+1}, \cdots, b_{2m} \}$. Now we make a key observation. $LEFT$ must be selected, otherwise any $b_i, b_j \in LOW$ cannot be compared. As long as $LEFT$ is selected, any two sets in $UP$ can be compared. Also, any two sets $b_i \in UP, b_j \in LOW$ can be compared if and only if $|i-j|\ne m$. The only pairs still needed to be compared are those $b_i, b_j$ where $|i-j|=m$. In order to compare $b_i$ with $b_{i+m}$, $1\le i \le m$, we have to make sure $b_i \cap RIGHT \ne \emptyset$. This is exactly selecting a subset of $RIGHT$ to hit each $b_i$, $1\le i \le m$. To do this, just select $k$ items in $\{a_{m+1}, \cdots, a_{m+n}\}$ which correspond to the set cover in $MSC(\mathcal{U},\mathcal{F},k)$.
\begin{table}[!ht]
  \centering
  \subfloat[$MSC(\mathcal{U},\mathcal{F},k)$ \label{tab:MSC(U,F,k)}]{
    \centering
\begin{tabular}{|c|c|c|c|}\hline
     & $f_1$ & $f_2$  & $f_3$ \\
    \hline
    $u_1$ & 1 & 1 & \\
    \hline
    $u_2$ &  & 1 & 1 \\
    \hline
    $u_3$ & 1 &  & 1 \\
    \hline
    $u_4$ & 1 &  &  \\
    \hline
    \end{tabular}
  }
  \qquad \qquad
  \subfloat[$MCS(A,B,|\mathcal{U}|+k)$\label{tab:MCS(A,B,|U|+k)}]{
    \centering
\begin{tabular}{|c|c|c|c|c|c|c|c|}\hline
     & $a_1$ & $a_2$  & $a_3$  & $a_4$ & $a_5$  & $a_6$  & $a_7$ \\
    \hline
    $b_1$ & 1 &  & & & 1 & 1 & \\
    \hline
    $b_2$ &  & 1 & & & & 1 & 1 \\
    \hline
    $b_3$ &  &  & 1 &  & 1 & & 1 \\
    \hline
    $b_4$ &  &  & & 1 &  1 & & \\
    \hline
    $b_5$ & 1 &  & & & & & \\
    \hline
    $b_6$ &  & 1 & & & & & \\
    \hline
    $b_7$ &  &  & 1 & & & & \\
    \hline
    $b_8$ &  &  & & 1 & & & \\
    \hline
    \end{tabular}
  }
\end{table}
\end{proof}

\subsection{Approximation and Inapproximability of MCS}
We first show that $MCS(A,B,k)$ is $O(\log{|B|})$ approximable as a reduction to Hitting Set Problem. Specifically, if $X$ is a compare set, then to compare any two subsets in $B$, say $b_i$ and $b_j$, $X \cap ( b_i\setminus b_j ) \ne \emptyset$ if $b_i\setminus b_j  \ne \emptyset$, and $X \cap ( b_j\setminus b_i ) \ne \emptyset$ if $b_j\setminus b_i  \ne \emptyset$. We will carry over the inapproximability of MSC to MCS by scaling the reduction used in the NP-Complete proof.


\begin{theorem}
The $MCS(A,B)$ has no polynomial-time algorithm with performance bound $o(\log{|B|})$ unless $\mathbf{P=NP}$. Also it has no polynomial-time algorithm with performance bound $(1-\epsilon)\ln{|B|}$, for any $\epsilon > 0$, unless $\mathbf{NP} \subset \mathbf{DTIME}(|B|^{\log{\log{|B|}}})$.
\end{theorem}

\begin{proof}
 For any $MSC(\mathcal{U},\mathcal{F})$, we make an $x$ times multiplied MSC instance with $x$ disjoint copies of it. Let $|\mathcal{U}|=m$, $|\mathcal{F}|=n$, then $x=\lceil{m\log m}\rceil$.  We refer to it as \emph{Mul-MSC$(x\mathcal{U}, x\mathcal{F})$}. We then construct an \emph{MCS} instance $MCS(x\mathcal{F}+\mathcal{U},(x+1)\mathcal{U} )$ by the same way used in the NP-Complete proof. Notation is abused here to make the idea clear. Recall the construction that, $MSC(\mathcal{U},\mathcal{F})$ has a solution of size at most $\delta$ if and only if  \emph{Mul-MSC$(x\mathcal{U}, x\mathcal{F})$} has a solution of size at most $x\delta$, if and only if $MCS(x\mathcal{F}+\mathcal{U},(x+1)\mathcal{U} )$ has a solution of size at most $x\delta+m \le x\delta (1+O(1/\log m))$. Suppose we could approximate $MCS(x\mathcal{F}+\mathcal{U},(x+1)\mathcal{U} )$ within a factor of $\rho$, then $MSC(\mathcal{U},\mathcal{F})$ has a ratio of $\rho(1+O(1/\log m)$. From the definition of $MCS(A,B)$ and the transformation above, we see that $m=|\mathcal{U}|=\frac{1}{2} \log B$. For $MCS(\mathcal{U}, \mathcal{F})$, \cite{Raz:1997:SEL:258533.258641} shows it has no polynomial-time algorithm with performance bound $o(\log{|\mathcal{U}|})$ unless $\mathbf{P=NP}$, \cite{Feige:1998:TLN:285055.285059} shows it has no polynomial-time algorithm with performance bound $(1-\epsilon)\ln{|\mathcal{U}|}$, for any $\epsilon > 0$, unless $\mathbf{NP} \subset \mathbf{DTIME}(|\mathcal{U}|^{\log{\log{|\mathcal{U}|}}})$. Thus the theorem is correct.
\end{proof}


Because each string in MSS corresponds to an item in MCS and vice versa, and each equivalence class of $U$ in MSS corresponds to each subset in MCS and vice versa. So it is straight forward to show the inapproximability of Minimum Spanning Set (MSS).
\begin{theorem}
Given a set $W$ of $m$-bit long binary strings, for the Minimum Spanning Subset (MSS) problem, there is no polynomial-time algorithm with performance bound $o(\log{m})$ unless $\mathbf{P=NP}$. Also there is no polynomial-time algorithm with performance bound $(1-\epsilon)\ln{m}$, for any $\epsilon > 0$, unless $\mathbf{NP} \subset \mathbf{DTIME}(m^{\log{\log{m}}})$ .
\end{theorem}

\section{Power of Negation}

\emph{Negation} is also an elementary operation on boolean value, so it is
helpful to see how the properties of the studied problems in the former
sections changed when \emph{negation} is allowed.

\begin{theorem}\label{lem:negationEqual}
Given $(W,s)$, $s$ is representable by $W$ where \emph{negation} is allowed if
and only if $s$ is representable by $W\cup \overline W$ where \emph{negation}
is not allowed.
\end{theorem}

\begin{proof}
First, if $s$ is representable by $W\cup \overline W$ where \emph{negation} is
not allowed, clearly $s$ is representable by $W$ where \emph{negation} is
allowed.

If $s$ is representable by $W$ where \emph{negation} is allowed, there is a
formula $\phi$ with operands in $W$ to represent $s$. Due to DeMorgan's laws,
there is an equivalent formula where each \emph{negation} operator only appears
immediately before some operand. Each pair of the \emph{negation} operator and
the operand immediately after it can be regarded as the operand in
$\overline{W}$. Therefore, $s$ is also representable by $W\cup \overline W$
where \emph{negation} is not allowed.
\end{proof}

According to the theorem, deciding whether $s$ is representable by $W$ when
\emph{negation} is allowed can also be solved by just inputting $(W
\cup\overline{W},s)$ to the algorithm in Section \ref{sec:decide_s}.

\begin{theorem}\label{thm:count_negation}
When \emph{negation} is allowed, the number of strings representable by $W$ is
$2^{|U|}$, where $U$ is the set of equivalence classes derived from $W\cup
\overline W$.
\end{theorem}

\begin{proof}
According to Theorem \ref{lem:negationEqual}, the number of representable
strings when \emph{negation} is allowed, equals to the number of strings
generated from $W \cup \overline{W}$ when \emph{negation} is forbidden.
Further, by Lemma \ref{lem:1Stringto1Set} and \ref{lem:1Setto1AChain}, the
number of strings representable by $W$ when \emph{negation} is allowed, is the
same as the number of antichains of $(U, \preceq_U)$.

Thus let us look at the structure of $U$ derived from $W\cup \overline W$.
Suppose $[i]\ne [j]$, then $T_i \neq T_j$. If $T_i \subseteq T_j$, then there exists $x$ such that
$x \in T_j$ and $x \notin T_i$. So $j \in \mathsf{Zero}(x)$ and $i \notin
\mathsf{Zero}(x)$. Thus $j \notin \mathsf{Zero}(\overline{x})$ and $i \in
\mathsf{Zero}(\overline{x})$. Since $\overline{x}$ is also in $W\cup \overline
W$, therefore $\overline{x} \in T_i$ and $\overline{x} \notin T_j$ which
contradict the premise that $T_i \subseteq T_j$. So $T_i \nsubseteq T_j$.
Similarly, $T_j \nsubseteq T_i$. Thus $[i] \neq [j]$ implies $[i]$ and $[j]$
are incomparable in $(U,\preceq_U)$. Therefore, any subset of $U$ makes up an
antichain of $(U,\preceq_U)$. So when \emph{negation} is allowed, the number of
strings representable by $W$ is $2^{|U|}$.
\end{proof}


\begin{theorem}
When \emph{negation} is allowed, the minimum representation subset problem is
still $\mathcal{NP}$-hard.
\end{theorem}

\begin{proof}
We reduce the \emph{minimum set cover problem} to it. Given an instance
$(\mathcal{U,S})$ where $\mathcal{U}=\{2,\dots,m\}$ and $\mathcal{S}=\{S_1,
\dots, S_n\}$. An instance of the minimum representation subset problem is
constructed as follows. Keep $m$ unchanged, let $W=\{s_i| \mathsf{Zero}(s_i)=
S_i\}$ (i.e., $\mathsf{One}(\overline{s_i})=S_i$) and $\mathsf{Zero}(s)=\{1\}$.
For a subfamily $\mathcal{C}$, $C$ denotes the corresponding set of strings,
and vice verse. If there is a subfamily $\mathcal{C}$ covers $\mathcal{U}$,
obviously $\bigvee_{x \in C}\overline{x} = s$. Conversely, if $s$ is
representable by a subset $C$ when \emph{negation} is allowed, then $s$ is
representable by $C \cup \overline C$ without \emph{negation} due to Theorem
\ref{lem:negationEqual}. Since $\forall x \in C$, $b^x_1 \neq 0$, according to
Theorem \ref{lem:basic}, $s=t_1=\bigvee_{x \in C}\overline{x}$. That is,
$\mathsf{One}(s)=\{2,3,\cdots,m\}=\bigcup_{x \in C}\mathsf{One}(\overline{x})$.
Thus $\bigcup_{X \in \mathcal{C}}{X} = \mathcal{U}$. So the minimum of
$(\mathcal{U,S})$ is the same as the minimum of $(W,s)$.
\end{proof}

\bibliographystyle{abbrv}
\bibliography{bstring}

\begin{thebibliography}{1}

\bibitem{Feige:1998:TLN:285055.285059}
U.~Feige.
\newblock A threshold of ln n for approximating set cover.
\newblock {\em J. ACM}, 45:634--652, July 1998.

\bibitem{Johnson:1973:AAC:800125.804034}
D.~S. Johnson.
\newblock Approximation algorithms for combinatorial problems.
\newblock In {\em Proceedings of the fifth annual ACM symposium on Theory of
  computing}, STOC '73, pages 38--49, New York, NY, USA, 1973. ACM.

\bibitem{karp1972}
R.~M. Karp.
\newblock {Reducibility Among Combinatorial Problems}.
\newblock In R.~E. Miller and J.~W. Thatcher, editors, {\em Complexity of
  Computer Computations}, pages 85--103. Plenum Press, 1972.

\bibitem{moret:983}
B.~M.~E. Moret and H.~D. Shapiro.
\newblock On minimizing a set of tests.
\newblock {\em SIAM Journal on Scientific and Statistical Computing},
  6(4):983--1003, 1985.

\bibitem{provan83}
S.~J. Provan and M.~O. Ball.
\newblock {The Complexity of Counting Cuts and of Computing the Probability
  that a Graph is Connected}.
\newblock {\em SIAM Journal on Computing}, 12(4):777--788, 1983.

\bibitem{Raz:1997:SEL:258533.258641}
R.~Raz and S.~Safra.
\newblock A sub-constant error-probability low-degree test, and a sub-constant
  error-probability pcp characterization of np.
\newblock In {\em Proceedings of the twenty-ninth annual ACM symposium on
  Theory of computing}, STOC '97, pages 475--484, New York, NY, USA, 1997. ACM.

\bibitem{Stone1936}
M.~H. Stone.
\newblock The theory of representation for boolean algebras.
\newblock {\em Transactions of the American Mathematical Society}, 40(1):pp.
  37--111, 1936.

\end{thebibliography}

\end{document}